\documentclass[a4paper, twocolumn]{article}
\usepackage[utf8]{inputenc}
\usepackage[english]{babel}

\usepackage{authblk}
\usepackage{amsfonts}
\usepackage{amssymb}
\usepackage{amsthm}
\usepackage{amsmath}

\newtheorem{theorem}{Theorem}[section]
\newtheorem{corollary}{Corollary}[theorem]

\newtheorem{definition}[theorem]{Definition}

\usepackage{braket}
\usepackage{verbatim}
\usepackage{listings}
\usepackage{rotating}
\usepackage{hhline}
\usepackage[table,x11names,xcdraw]{xcolor} 
\definecolor{vlightgray}{cmyk}{0,0,0,0.125}

\usepackage{tikz}
\usetikzlibrary{arrows.meta}
\usetikzlibrary{quantikz}

\usepackage[letterpaper,top=2cm,bottom=2cm,left=3cm,right=3cm,marginparwidth=1.75cm]{geometry}

\usepackage{lscape}
\usepackage{amsmath}
\usepackage{graphicx}
\usepackage{subcaption}
\usepackage{mwe}
\usepackage[colorlinks=true, allcolors=blue]{hyperref}

\title{Cutting Medusa's Path \\
    \large Tackling Kill-Chains with \\ Quantum Computing}
\author[1]{Mark Carney}
\affil[1]{Quantum Village Inc.}
\date{November 2022}

\begin{document}
\twocolumn[
\begin{@twocolumnfalse}
\maketitle

\begin{abstract}
This paper embarks upon exploration of quantum vulnerability analysis. By introducing vulnerability graphs, related to attack graphs, this paper provides background theory and a subsequent method for solving significant cybersecurity problems with quantum computing. The example given is to prioritize patches by expressing the connectivity of various vulnerabilities on a network with a QUBO and then solving this with quantum annealing. Such a solution is then proved to remove all kill-chains (paths to security compromise) on a network. The results demonstrate that the quantum computer's solve time is almost constant compared to the exponential increase in classical solve time for vulnerability graphs of expected real world density. As such, this paper presents a novel example of advantageous quantum vulnerability analysis.
\end{abstract}
\end{@twocolumnfalse}]

\section{Introduction}

Much of the defensive work carried out by organisations in the cybersecurity space has become a data problem \cite{scott_2021}. Whilst this presents several challenges, it also presents numerous opportunities for moving the operating mode of any Security Operations Centre (SOC) from mostly reactive to including predictive activities \cite{Chio2018}.

Quantum computing has heralded several major claims for the future of computing, particularly surrounding the significant optimisations (theoretically) to be found in quantum computing on specific problem classes, including number factoring, parameter discovery, travelling salesman problems, and many more \cite{Nielsen2010}.

This paper presents a method for utilising such quantum advantage in a new approach to the theoretical underpinnings of network-based cybersecurity. By finding a new way of solving some of the data problems in cyber-defence, and improving the response times thereby, the aim is to demonstrate a more efficient move towards prioritising and applying vulnerability patches.

Patch management is a common pain point for any large scaled enterprises \cite{NIST2022} or widely distributed systems such as smartphones or IoT devices \cite{Farhang2019}. Indeed, lack of appropriate patching was indicated as being a central cause for some high profile cybersecurity breaches, such as the infamous Equifax hack \cite{newman2017}. A variety of approaches have been proposed to improve the categorization and management of patches, including deep learning technologies \cite{Wang2021}.

This paper demonstrates this by analysing vulnerability data on hosts as a bipartite graph. With this, we reason that attacks are made up of `kill-chains', which are themselves comprised of sequences of exploits leveraging vulnerabilities that are coincident - in our model, by being shared on a host. By then creating a process through which we can totally disconnect vulnerabilities from one another, we effectively remove every possible kill-chain. This problem, however, involves a known \emph{NP}-hard problem, but one which is tractable on a DWave quantum computer (see section \ref{sec:GA}). The full results of various tests of this method can be found in section \ref{sec:results}, and the code may be found in \cite{CarneyGH2022}.

By leveraging quantum computation and optimisation methods for vulnerability analysis of this kind, new avenues of optimisation of cybersecurity and related data can be considered. The example presented here is fully worked through to demonstrate where a possible quantum advantage may lie for this kind of analysis on this type of data.

The leveraging of short-time solutions to \emph{NP}-hard problems that are present in cybersecurity data is a potentially rich vein of exciting possibilities. The fast and efficient resolution of cybersecurity data problems also helps reduce the analysis and reaction times of security teams, thereby tightening their OODA loop and increasing their security level proportionately \cite{Schwartau2018}.

\section{Prerequisites}

\subsection{Graph Theoretic Pre-requisites}

\begin{definition}{\cite{Diestel2010}}
A \emph{graph} $G = (V,E)$; $V(G)$ a set of vertices and a set of edges $E(G) \subseteq V \times V$, $E$ is composed of pairs of elements from $V$.
\end{definition}

All of the graphs we will consider here are finite simple graphs; the number of edges and vertices is finite, and the graph has only one unique edge between any two vertices, with no edges starting and terminating on the same vertex (simple).

\begin{definition}{\cite{Diestel2010}}
A \emph{path} $P$ is a sequence of edges $e_0, e_1, \ldots, e_n \in E$ such that for every $e_{i-1}, e_i, e_{i+1} \in P$, $\exists v_i \in (e_{i-1} \cap e_i) \text{ and } \exists v_{i+1} \in (e_{i} \cap e_{i+1})$. There are two distinguished vertices, one at the start and one at the end, and a path $P$ is a \emph{loop} iff the starting and terminating edges in the path overlap.
\end{definition}

\begin{definition}{\cite{Diestel2010}}
    A \emph{weighted} graph $G$ is a graph with edges that are triples $(u,v,w)$ such that $u,v \in V$ and weight $w \in \mathbb{N}$. If $w$ is not specified, it is assumed to be 1.
\end{definition}

The use of the weight parameter $w$ is such that we can talk about the maximum/minimum `flow' across different paths along a graph.

We shall be using a bipartite graph structure, defined as follows:

\begin{definition}{\cite{Diestel2010}}
    A bipartite graph $G$ is a graph with a partition of $V(G)$ into two sets $A, B$ such that $\forall (a,b) \in E(G)$, $a \in A$ and $b \in B$.
\end{definition}

\subsection{Graph Algorithms}\label{sec:GA}

Various algorithms have been found that provide graph structured information with significant insight into the underlying relationships between different elements in a dataset. The graph algorithm utilised in this paper concerns finding the \emph{minimum vertex cover}. 

\begin{definition}[MVC, \cite{Karp1972}]
    A \emph{vertex cover} for a graph $G$ is a subset $C \subset V(G)$ such that for every edge $(u,v) \in E(G)$ at least one of $u$ or $v$ is in $C$. $C$ is \emph{minimal} if $V(G) \setminus C$ is as large as possible.
\end{definition}

The idea behind a vertex cover is that it is a set of vertices that `touches' every edge in the graph. That is for any vertex cover $C$ for a graph $G$, $V(G) \setminus C \implies E(G) = \emptyset $. An important feature of vertex covers is that finding one is a known \emph{NP}-hard problem.\cite{Karp1972}

Let an \emph{independent set} $\mathcal{I}$ be a set of vertices on a graph such that no two members of $\mathcal{I}$ are connected. An alternative definition that we use in the code in \cite{CarneyGH2022} is that a minimum vector cover of some graph $G$ is the compliment of a maximal independent set $\mathcal{I} \subseteq V(G)$.

\subsection{Quantum Annealing and QUBOs for Graph Algorithms}

There are many problems that have been found to be easily coded into Quadratic Unconstrained Binary Optimisation (QUBO) problems\footnote{Also called `Unconstrained Binary Quadratic Programs', or UBQPs. See \cite{Kochenberger2014}.}\cite{Lucas2014}. Indeed, DeSimone \emph{et al.} \cite{DeSimone1995} showed how Ising Hamiltonians can be equivalent to graph problems, an idea that we make use of here.

\begin{definition}{\cite{DeSimone1995}}\label{defQUBO}
    A \emph{Quadratic Unconstrained Binary Optimisation} (QUBO) is the problem of finding, for $x \in \mathbb{B}^n$ the following $$ \min x^\top Q x$$ for a given upper triangular matrix $Q \in \mathbb{R}^{n \times n}$. 
\end{definition}

DeSimone \emph{et al.} \cite{DeSimone1995} showed that QUBO problems can be represented in the form $$ f_Q(x) = \sum_{i=1}^n \sum_{j=1}^i q_{ij} x_i x_j $$ for coefficients $q_{ij} \in \mathbb{R}$, $x_i, x_j \in \mathbb{B}^n$. This is equivalent to extracting the coefficients from $Q$ in definition \ref{defQUBO}. This is equivalent to Ising Hamiltonians of the form $$ H(\sigma) = -\sum_{\langle i,j \rangle} J_{ij} \sigma_i \sigma_j - \mu \sum_{j} \sigma_0 \sigma_j $$ for real parameters $J_{ij}$, $\sigma_i$, and $\mu$. 

With this equivalence, there is a way to use the adiabatic quantum computing methods of DWave Systems to acquire the vector values for $x$ above such that it can solve the minimum vector cover problem we described in section \ref{sec:GA} \cite{Pelofske2019}.

\section{New Approaches to Vulnerability Patch Prioritisation}

This section describes a graph theoretic way of analysing vulnerabilities on a computer network called `vulnerability graphs', derived from the notion of `attack graphs' \cite{Matthews2020}. The aim of this section is to present a new way to prioritise the patching of vulnerabilities by considering their connectivity and solving for this using quantum computation.

Our aim here is to provide a way of identifying the `most well connected' issues by means of describing the dual on our vulnerability graphs. 

\subsection{Attack Graphs and Kill-chains}

`Attack graphs' have featured in some interesting approaches to managing and mitigating security threats (\cite{FX2016}, \cite{Callahan2019}, \cite{Matthews2020}). They provide various ways regarding how to analyse network-oriented vulnerability data that many cybersecurity information sources generate. As pointed out in Fran\c{c}ois-Xavier \emph{et al.} \cite{FX2016}, representations of a Bayesian Attack Model can take the form of network graphs, from which Matthews \emph{et al.} \cite{Matthews2020} show that various network characteristics can be extracted. Callahan \emph{eg al.} go so far as to construct and optimisation model for virtualised networks that incorporate attack graphs to give some guarantee of security.

Whilst the formulation in \cite{Matthews2020} utilises directed graphs, this paper uses un-directed simple graphs to represent the same data. This is due to the fact that it is not that important to consider directionality for the purposes of this analysis.

A `kill chain' is a multi-stage sequence of events that leads to the compromise of a network \cite{Yadav2015}. What is apparent is that many of the examples of kill-chains involve sequences of vulnerabilities, with the sequence dependant on the assets that are intersected between these vulnerabilities.

\subsection{Vulnerability Graphs}

Define a vulnerability graph as follows:
\begin{definition}
    A \emph{vulnerability graph} $G$ is a bipartite graph where one partition of vertices represents network hosts, and the other represents vulnerabilities. Let the edges of $G$ represent that a given host is affected by some detected vulnerability.
\end{definition}

This data structure is a very direct way of representing the results from industry standard vulnerability scanning and penetration test reports \cite{scott_2021}. This particular data structure lends itself to analysing the data in vulnerability reports much more efficiently. This is due to the fact that vulnerability reports are asset-first, \emph{i.e.} they are a list of hosts with sub-lists of what vulnerabilities affect them. However, many of the questions we wish to ask of this dataset are the other way round - we wish for a list of vulnerabilities with sub-lists of affected hosts. To search host-first data to extract all vulnerabilities, it is clear that all the host records must be read. This is computationally expensive, compared with a graph structure.

For this paper, we shall consider kill-chains as sub-sequences of paths through the vulnerability graph. For our purposes, we shall use the following definition:
\begin{definition}\label{def:killchain}
    A \emph{kill chain} is a sequence of vertices $K = \{ v_1, v_2, \ldots v_n \}$ from the vulnerability partition of a vulnerability graph $\mathcal{V}$ such that for each $v_i, v_j \in K$, there exists at least one host $h \in \mathcal{V}$ with $(v_i,h), (v_j,h) \in E(\mathcal{V})$. 
\end{definition}

Each part of a kill chain that comprises of movement from one vulnerability to the next will start on some vulnerability vertex, go to some host vertex, and then on to another vulnerability vertex connected to that same host. 

This paper shall only consider kill-chains that require vulnerabilities to be on the same host, not adjacent hosts on a given network subnet. However, it is very easy to consider adding nodes to the host partition of the vulnerability graph to encode these visibility relationships for vulnerabilities that are network facing. But to save space, the approach here will not consider this for now.

It is worth pointing out that the lack of any information coded about severity ratings for vulnerabilities, \emph{e.g.} CVSS scores. This information is not considered here as critical vulnerabilities should always be patched as soon as possible. The aim of these definitions is to consider non-critical paths to compromise as generally a subset of all paths through a vulnerability graph, and solve for these issues by analysing their connectivity. 

In short, this process aims to find the lower criticality issues that are widespread and well connected enough to cause more harm later.

It is important not to only consider only the vulnerabilities that we know to be `very bad'. This approach, in line with the quantum speedup we describe later, allows cybersecurity to consider an entire attack surface, not just potentially isolated pain points. This is the logical extension of the common maxim ``defence in depth" \cite{Groat2012}.

\subsection{Connectivity Dual Graphs and Vector Covers}\label{sec:dual}

We define the following dual graph construction for a vulnerability graph $\mathcal{V}$ with a partition of vulnerability vertices and host vertices:
\begin{definition}
    The dual $\mathcal{D}_{\mathcal{V}}$ is constructed as follows. For each vulnerability vertex $v_i \in V(\mathcal{V})$ for $1 \leq i \leq |V(\mathcal{V})|$:
    \begin{enumerate}
        \item Add $v_i$ to $V(\mathcal{D}_{\mathcal{V}})$ if $v_i \notin V(\mathcal{D}_{\mathcal{V}})$.
        \item Enumerate a list of host vertices $\{ h_0 , h_1 , \ldots \}$ connected to $v_i$.
        \item Iterating over this list of hosts, for each $v'_i$ connected to each host $h_j$:
        \begin{itemize}
            \item Add $v'_i$ to $V(\mathcal{D}_{\mathcal{V}})$.
            \item Add $(v_i, v'_i)$ to $E(\mathcal{D}_{\mathcal{V}})$
            \item If $(v_i, v'_i)$ already exists, add 1 to the weight of that edge.
        \end{itemize}
        \item Remove $v_i$ from $V(\mathcal{V})$ and continue with $v_{i+1}$.
    \end{enumerate}
\end{definition}

The dual $\mathcal{D}_{\mathcal{V}}$ represents all of the connections between vulnerabilities on a network. It is weighted, so that priority can be given to vulnerabilities that are more connected than others, by virtue of the weight coding the number of hosts a vulnerability was found on.

\subsection{Removing Kill-chains with Vertex Covers}

Previous approaches, notably \cite{Matthews2020}, have relied on locating cycles of probabilities in attack graphs. The idea in this paper is to utilise vulnerability graphs to consolidate all possible avenues for compromise, the kill-chains, and then remove all possible attack routes through a network. 

\begin{theorem}\label{thm:main}
    Removing the vertices in a vertex cover on $\mathcal{D}_{\mathcal{V}}$ from $V(\mathcal{V})$ will leave $\mathcal{V}$ totally disconnected on the vulnerability partition to itself via the host partition.
\end{theorem}

\begin{proof}
    Edges on the dual graph $\mathcal{D}_{\mathcal{V}}$ represent an `edge-host-edge' sub-path between one vulnerability and another on $\mathcal{V}$. The vertex cover on $\mathcal{D}_{\mathcal{V}}$ therefore intersects every `vulnerability-host-vulnerability' path on $\mathcal{V}$. 
    
    Let $v_i, v_j$ be vulnerabilities and $h_k$ a host in $V(\mathcal{V})$. Removing every host in the vertex cover of $\mathcal{D}_{\mathcal{V}}$ will result in a path of the form $\{(v_i,h_k),(h_k,v_j)\}$ being removed, and the sequence $\{v_i, h_k, v_j\}$ reduced to just one of $\{v_i, h_k\}$ or $\{h_k,v_j\}$. Therefore, when all the vertices are removed, $\mathcal{V}$ will be disconnected from the vulnerability side, as there is no way to get from any vulnerability to any other vulnerability.
\end{proof}

The point of this proof is to show that in order to remove all possible kill-chains, we do not need to resolve all vulnerabilities. 

\begin{corollary}
    A disconnected vulnerability graph can contain no kill-chains.
\end{corollary}
\begin{proof}
    Consider our definition \ref{def:killchain}. By removing every `vulnerability-host-vulnerability' sub-path in a vulnerability graph $\mathcal{V}$ by means of a minimum vertex cover on $\mathcal{D}_{\mathcal{V}}$, we have removed every possible kill chain $K$ found in the paths of $\mathcal{V}$.
\end{proof}

Based on our assumption of a kill chain's structure, this approach will remove every kill chain. Furthermore, it becomes clear that we do not need to patch everything in order to have a significant impact and improvement. We can patch smarter, not harder. 

\subsection{Considerations}


It has been a mainstay for some time that many notions of security rely, in some significant sense, on the reaction times of those tasked with defending assets \cite{Schwartau2018}. Simply put; the shorter our detection and reaction times and the faster we iterate over security data with feedback from new information, the better our security posture will be.

The problem with considering attack graphs at enterprise scale is the sheer size of the datasets that could be involved. Thus, if we can minimise time taken to analyse large vulnerability datasets then we can improve enterprise reaction times to these issues being identified. 

By utilising an adiabatic quantum computing setup, this time should be shortened as we have solved an \emph{NP}-hard problem efficiently. By adding in other work by Pelofske \emph{et al} \cite{Pelofske2019} we can reduce the time required to find the most at-risk vulnerabilities and patch them with more priority. 

However, as pointed out in Di Tizio \emph{et al} \cite{DiTizio2022}, much of the activity in patch management may be relatively superfluous for any attack that is below the complexity of a nation state APT level threat actor.


%
%

\subsection{Weighted MVC}

We recall that in the construction process of the dual in section \ref{sec:dual}. As pointed out in \cite{Karp1972}, the weighted minimum vector cover (wMVC) problem is reducible to the non-weighted MVC. As shown in Pelofske \emph{et al} \cite{Pelofske2019}, these can also be performed on adiabatic quantum computers by means of coding the problem into a QUBO. 

Therefore, by applying a weighted MVC to our weighted dual $\mathcal{D}_{\mathcal{V}}$, we can improve our output by being able to prioritise more highly connected vulnerabilities over less well connected ones. Although this does not change any of the proofs above, it is a significant improvement our analysis. 

\section{A Worked Example}

We present the following vulnerability graph $\mathcal{V}$, with hosts $a$ to $g$, and vulnerabilities $1$ through $8$:

\begin{center}
\scalebox{0.51}{\begin{tikzpicture}
[nodeDecorate/.style={inner sep=2pt,draw,thick,minimum size=0.4cm},%
  lineDecorate/.style={-,thick}]
\foreach \nodename/\x/\y/\c/\s in {
  1/0/0/blue/circle, 2/2/0/blue/circle, 3/4/0/blue/circle, 4/6/0/blue/circle, 5/8/0/blue/circle, 6/10/0/blue/circle, 7/12/0/blue/circle, 8/14/0/blue/circle,
  a/1/2.5/red/rectangle, b/3/2.5/red/rectangle, c/5/2.5/red/rectangle, d/7/2.5/red/rectangle, e/9/2.5/red/rectangle, f/11/2.5/red/rectangle, g/13/2.5/red/rectangle}
{
  \node (\nodename) at (\x,\y) [nodeDecorate] [label=left: \nodename,fill=\c,shape=\s] {};
}
\path
\foreach \startnode/\endnode in {
  1/a, 1/b, 1/d, 1/f,
  2/a, 2/b,
  3/a, 3/d, 3/e,
  4/b, 4/c, 4/f,
  5/g, 6/f, 6/g,
  7/b, 7/c, 7/f,
  8/c, 8/d, 8/g}
{
  (\startnode) edge[lineDecorate] node {} (\endnode)
};
\end{tikzpicture}}
\end{center}

The host connections are given by the following vulnerability to host edges, summarised by:

\begin{itemize}
    \item $1$ --- $\{ a, b, d, f \}$
    \item $2$ --- $\{ a, b \}$
    \item $3$ --- $\{ a, d, e \}$
    \item $4$ --- $\{ b, c, f \}$
    \item $5$ --- $\{ g \}$
    \item $6$ --- $\{ f, g \}$
    \item $7$ --- $\{ b, c, f \}$
    \item $8$ --- $\{ c, d, g \}$
\end{itemize}
From this graph we can get the following dual graph $\mathcal{D}_{\mathcal{V}}$:

\begin{center}
\begin{tikzpicture}
[nodeDecorate/.style={inner sep=2pt,draw,thick,minimum size=0.5cm,shape=circle},%
  lineDecorate/.style={-,thick}]
\foreach \name/\x/\y in {
1/2/0, 2/4/0, 3/6/2, 4/6/4, 5/4/6, 6/2/6, 7/0/4, 8/0/2
}
{
  \node (\name) at (\x,\y) [nodeDecorate] {$\name$};
}
\path
\foreach \startnode/\endnode in {
  1/2,1/3,1/4,1/7,1/8,1/6,
  2/3,2/4,2/7,
  3/8,
  4/6,4/7,4/8,
  5/6,5/8,
  6/7,6/8,
  7/8}
{
  (\startnode) edge[lineDecorate] node {} (\endnode)
};

\end{tikzpicture}
\end{center}

If we compute the MVC we get the set $\{ 1, 2, 4, 8, 6 \}$. Removing these nodes from our original vulnerability graph $\mathcal{V}$ we get: 

\begin{itemize}
    \item $3$ --- $\{ a, d, e \}$
    \item $5$ --- $\{ g \}$
    \item $7$ --- $\{ b, c, f \}$
\end{itemize}

Or as a diagram:
\begin{center}
\scalebox{0.58}{\begin{tikzpicture}
[nodeDecorate/.style={shape=circle,inner sep=2pt,draw,thick,minimum size=0.4cm},%
  lineDecorate/.style={-,thick}]
\foreach \nodename/\x/\y/\c/\s in {
  3/8/0/blue/circle, 5/12/0/blue/circle,  7/4/0/blue/circle,
  a/0/2.5/red/rectangle, b/2/2.5/red/rectangle, c/4/2.5/red/rectangle, d/6/2.5/red/rectangle, e/8/2.5/red/rectangle, f/10/2.5/red/rectangle, g/12/2.5/red/rectangle}
{
  \node (\nodename) at (\x,\y) [nodeDecorate] [label=left: \nodename,fill=\c,shape=\s] {};
}
\path
\foreach \startnode/\endnode in {
  3/a, 3/d, 3/e, 5/g, 7/b, 7/c, 7/f}
{
  (\startnode) edge[lineDecorate] node {} (\endnode)
};
\end{tikzpicture}}
\end{center}
This demonstrates that our algorithm and theorem is correct, as well as illustrating the density which can arise in vulnerability graphs and the duals we defined for them. 

\section{Solving with Quantum Hardware}

We now demonstrate our approach using a DWave AQC system.

\subsection{Implementation}

To test the above process on quantum hardware, it was decided to make use of the built in functions within DWave's implementation of the \verb'networkx' python library for graph programming. The back-end that was used was the `Advantage Solver 4.1' QPU system.

For benchmarking, the same algorithms were run through the accompanying \verb'ExactSolver()' for QUBO's that solves using classical methods. We implemented our own solver, based on DWave's provided code, for finding independent sets with DWave's python interface. This gave us more control over the generation of solutions and checking. The code is available on our github: see \cite{CarneyGH2022}.

To explore the hardware capabilities, python code was written to perform 3 solves on the duals automatically generated for increasingly larger random bipartite graphs with varying edge probabilities, timing the results on both classical and quantum solvers \cite{CarneyGH2022}. The given MVC candidates are then checked that all `vuln-host-vuln' paths have been eliminated.

\begin{table*}[t]
\centering
\begin{tabular}{l|rr|}
\cline{2-3}
& \multicolumn{2}{c|}{\cellcolor[HTML]{EFEFEF}\textbf{\begin{tabular}[c]{@{}c@{}} \hline Number of vulnerabilities\\ Disclosed to CVE Databases\end{tabular}}} \\ \hline
\rowcolor[HTML]{EFEFEF} 
\multicolumn{1}{|l|}{\cellcolor[HTML]{EFEFEF}\textbf{Software Packages}} & \multicolumn{1}{r|}{\cellcolor[HTML]{EFEFEF}\textbf{2020}}                                    & \textbf{All Time (as of 2021)}                                    \\ \hline
\multicolumn{1}{|l|}{Windows 10}                                         & \multicolumn{1}{r|}{807}                  & 2990 \\ \hline
\multicolumn{1}{|l|}{Windows Server 2016}                                & \multicolumn{1}{r|}{794}                  & 2764 \\ \hline
\multicolumn{1}{|l|}{Linux Kernel (all versions)}                        & \multicolumn{1}{r|}{126}                  & 3000 \\ \hline
\multicolumn{1}{|l|}{Debian Linux}                                       & \multicolumn{1}{r|}{946}                  & 7331 \\ \hline
\multicolumn{1}{|l|}{Ubuntu Linux}                                       & \multicolumn{1}{r|}{483}                  & 3680 \\ \hline
\multicolumn{1}{|l|}{Mac OSX (all versions)}                             & \multicolumn{1}{r|}{314}                  & 3100 \\ \hline
\multicolumn{1}{|l|}{Android (all versions)}                             & \multicolumn{1}{r|}{859}                  & 4707 \\ \hline
\multicolumn{1}{|l|}{iOS (for iPhone)}                                   & \multicolumn{1}{r|}{322}                  & 2820 \\ \hline
\multicolumn{1}{|l|}{Chrome Browser}                                     & \multicolumn{1}{r|}{227}                  & 2554 \\ \hline
\multicolumn{1}{|l|}{Firefox Browser}                                    & \multicolumn{1}{r|}{141}                  & 1993 \\ \hline
\multicolumn{1}{|l|}{Microsoft Office Suite}                             & \multicolumn{1}{r|}{71}                   & 727  \\ \hline
\end{tabular}
\caption{This table shows the number of disclosed vulnerabilities to the CVE database \cite{Ozkan22} for some select popular software packages, for 2020 specifically and `All Time'.}
\label{tab:vulnspersys}
\end{table*}

\subsection{Limitations and Real World Patching Estimates}

Owing to the lack of publicly available vulnerability data from scans that consisted of more than a handful of hosts, it was not possible to validate this approach using real world data.

However, one can estimate the possible occurrence of patches for a given host based on the number of vulnerabilities disclosed for major software operating systems, components, and suites. We use \cite{Ozkan22} as our primary resource for this data, a summary of some examples is in table \ref{tab:vulnspersys}.

For any vulnerability to enter the CVE database \cite{Ozkan22}, it is usually accompanied by some kind of coordinated disclosure, for which a patch is generally also made available. As such, the number of publicly disclosed vulnerabilities should correlate reasonably accurately to the number of available patches for that software. 

Considering our chart in table \ref{tab:vulnspersys} we find that anywhere from 7-27\% of all vulnerabilities for a given software package were disclosed in 2020 alone. We can estimate the average annual patch inflation rate (that is, growth in the number of patches year on year) to be somewhere between 5-15\%.

Although there is a lack of any hard evidence for how many clients and servers any given enterprise commissions, maintains, or decomissions in any year, we can estimate from our annual patch inflation rate that for a given 5 year period the number of patches will grow between 27\% and 100\% for our patch inflation range. 

If we assume largely homogeneous networks - \emph{i.e.} that most hosts are derived from a small number of `gold builds' to facilitate the recommendations in \cite{NIST2022} - then the number of vulnerabilities  will be relatively small compared to the overall number of hosts. 

Combining all of these estimates, we find that with an average (from \cite{Ozkan22}) of $\approx 3000$ vulnerabilities per major software package, an enterprise with 10 dominant software packages will accrue anywhere up to 8,000 new vulnerabilities in 5 years with a patch inflation rate of 5\%, each package potentially adding $\approx 800$ vulnerabilities to the CVE database. 

For an enterprise with 100,000 hosts (clients and servers) this would give a minimal guesstimate of $\approx 8\%$ probability that any vulnerability is connected to any given host, assuming an even distribution of software. This is the estimate that we will use later in our experimentation with a DWave AQC, but should be considered somewhat of an upper bound.

Clearly time really is of the essence when applying patches, and so any process to speed this activity up has significant potential utility given the potential growth rates involved.

\subsection{Results and Analysis}\label{sec:results}

\begin{figure*}[t]
    \centering
    \begin{subfigure}[b]{0.475\textwidth}
        \centering
        \includegraphics[width=\textwidth]{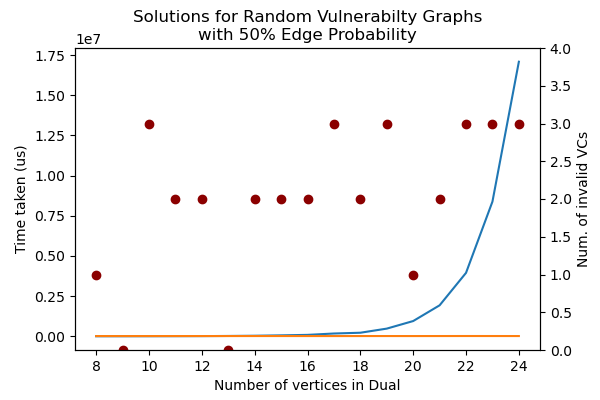}
    \end{subfigure}
    \hfill
    \begin{subfigure}[b]{0.475\textwidth}
        \centering
        \includegraphics[width=\textwidth]{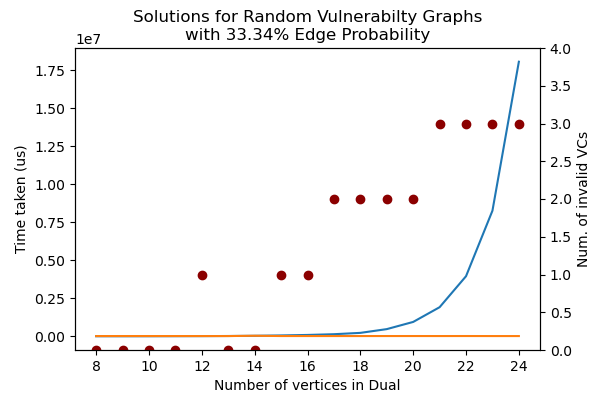}
    \end{subfigure}
    \hfill
    \begin{subfigure}[b]{0.475\textwidth}
        \centering
        \includegraphics[width=\textwidth]{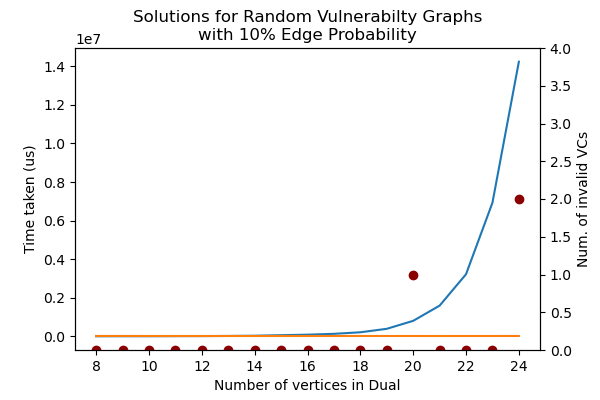}
    \end{subfigure}
    \hfill
    \begin{subfigure}[b]{0.475\textwidth}
        \centering
        \includegraphics[width=\textwidth]{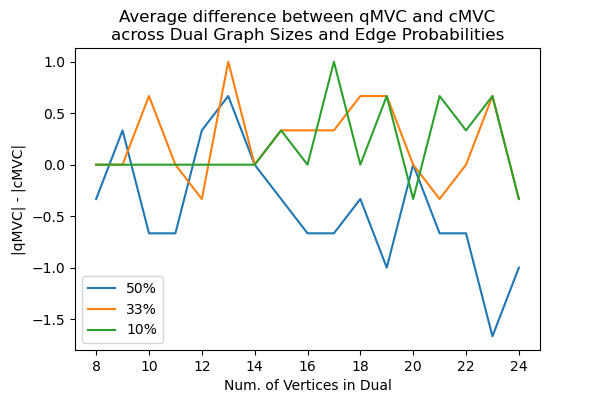}
    \end{subfigure}
    \hfill
    \caption{The results of running the outlined procedure in section \ref{sec:dual} on three random graphs of varying edge probabilities, ranging in sizes 16 to 48 vertices (8 to 24 vulnerabilities). Edge probabilities determine the likelihood that any two vertices have an edge between them. Times for the quantum annealer are shown by the orange line, with the blue line indicating the time taken to find the MVC clasically, in microseconds (1e7 scale). The red dots indicate the number of invalid vector cover solutions found by the DWave annealer. The lower right graph shows the averaged difference in sizes between the classical method for finding MVCs (denoted cMVC), and the DWave AQC (deonted qMVC).}
    \label{fig:results}
\end{figure*}

\begin{table*}[t]
\resizebox{\textwidth}{!}{\begin{tabular}{l|llll|llll|llll|}
\cline{2-13}
  & \multicolumn{4}{c|}{\textbf{50\% Edge Probability}}
  & \multicolumn{4}{c|}{\textbf{33.34\% Edge Probability}}
  & \multicolumn{4}{c|}{\textbf{10\% Edge Probability}}                   \\ \hline
\multicolumn{1}{|l|}{\textbf{Num. Vulns}} & \multicolumn{1}{l|}{\begin{tabular}[c]{@{}l@{}}Mean Diff of\\ VC sizes\end{tabular}} & \multicolumn{1}{l|}{cMVC ($\mu s$)} & \multicolumn{1}{l|}{qMVC ($\mu s$)} & \begin{tabular}[c]{@{}l@{}}Invalid \\ VCs\end{tabular} & \multicolumn{1}{l|}{\begin{tabular}[c]{@{}l@{}}Mean Diff of\\ VC sizes\end{tabular}} & \multicolumn{1}{l|}{cMVC ($\mu s$)} & \multicolumn{1}{l|}{qMVC ($\mu s$)} & \begin{tabular}[c]{@{}l@{}}Invalid \\ VCs\end{tabular} & \multicolumn{1}{l|}{\begin{tabular}[c]{@{}l@{}}Mean Diff of\\ VC sizes\end{tabular}} & \multicolumn{1}{l|}{cMVC ($\mu s$)} & \multicolumn{1}{l|}{qMVC ($\mu s$)} & \begin{tabular}[c]{@{}l@{}}Invalid \\ VCs\end{tabular} \\ \hline
\multicolumn{1}{|l|}{\textbf{8}}          & \multicolumn{1}{l|}{-0.333}      & \multicolumn{1}{l|}{2847}           & \multicolumn{1}{l|}{16435}          & 1  & \multicolumn{1}{l|}{0}           & \multicolumn{1}{l|}{3677}           & \multicolumn{1}{l|}{16435}          & 0  & \multicolumn{1}{l|}{0}           & \multicolumn{1}{l|}{1911}           & \multicolumn{1}{l|}{16982}          & 0  \\ \hline
\multicolumn{1}{|l|}{\textbf{9}}          & \multicolumn{1}{l|}{0.333}       & \multicolumn{1}{l|}{3450}           & \multicolumn{1}{l|}{15579}          & 0  & \multicolumn{1}{l|}{0}           & \multicolumn{1}{l|}{2485}           & \multicolumn{1}{l|}{17851}          & 0  & \multicolumn{1}{l|}{0}           & \multicolumn{1}{l|}{2734}           & \multicolumn{1}{l|}{17395}          & 0  \\ \hline
\multicolumn{1}{|l|}{\textbf{10}}         & \multicolumn{1}{l|}{-0.667}      & \multicolumn{1}{l|}{4876}           & \multicolumn{1}{l|}{16451}          & 3  & \multicolumn{1}{l|}{0.667}       & \multicolumn{1}{l|}{3812}           & \multicolumn{1}{l|}{17547}          & 0  & \multicolumn{1}{l|}{0}           & \multicolumn{1}{l|}{3385}           & \multicolumn{1}{l|}{16150}          & 0  \\ \hline
\multicolumn{1}{|l|}{\textbf{11}}         & \multicolumn{1}{l|}{-0.667}      & \multicolumn{1}{l|}{8156}           & \multicolumn{1}{l|}{17095}          & 2  & \multicolumn{1}{l|}{0}           & \multicolumn{1}{l|}{7738}           & \multicolumn{1}{l|}{17464}          & 0  & \multicolumn{1}{l|}{0}           & \multicolumn{1}{l|}{8022}           & \multicolumn{1}{l|}{15646}          & 0  \\ \hline
\multicolumn{1}{|l|}{\textbf{12}}         & \multicolumn{1}{l|}{0.333}       & \multicolumn{1}{l|}{12256}          & \multicolumn{1}{l|}{16765}          & 2  & \multicolumn{1}{l|}{-0.333}      & \multicolumn{1}{l|}{11649}          & \multicolumn{1}{l|}{16582}          & 1  & \multicolumn{1}{l|}{0}           & \multicolumn{1}{l|}{11038}          & \multicolumn{1}{l|}{16214}          & 0  \\ \hline
\multicolumn{1}{|l|}{\textbf{13}}         & \multicolumn{1}{l|}{0.667}       & \multicolumn{1}{l|}{22217}          & \multicolumn{1}{l|}{16413}          & 0  & \multicolumn{1}{l|}{1}           & \multicolumn{1}{l|}{20335}          & \multicolumn{1}{l|}{16505}          & 0  & \multicolumn{1}{l|}{0}           & \multicolumn{1}{l|}{18499}          & \multicolumn{1}{l|}{16114}          & 0  \\ \hline
\multicolumn{1}{|l|}{\textbf{14}}         & \multicolumn{1}{l|}{0}           & \multicolumn{1}{l|}{36199}          & \multicolumn{1}{l|}{17154}          & 2  & \multicolumn{1}{l|}{0}           & \multicolumn{1}{l|}{33898}          & \multicolumn{1}{l|}{36734}          & 0  & \multicolumn{1}{l|}{0}           & \multicolumn{1}{l|}{29909}          & \multicolumn{1}{l|}{16129}          & 0  \\ \hline
\multicolumn{1}{|l|}{\textbf{15}}         & \multicolumn{1}{l|}{-0.333}      & \multicolumn{1}{l|}{60240}          & \multicolumn{1}{l|}{17238}          & 2  & \multicolumn{1}{l|}{0.333}       & \multicolumn{1}{l|}{57864}          & \multicolumn{1}{l|}{17940}          & 1  & \multicolumn{1}{l|}{0.333}       & \multicolumn{1}{l|}{57489}          & \multicolumn{1}{l|}{16195}          & 0  \\ \hline
\multicolumn{1}{|l|}{\textbf{16}}         & \multicolumn{1}{l|}{-0.667}      & \multicolumn{1}{l|}{89672}          & \multicolumn{1}{l|}{16380}          & 2  & \multicolumn{1}{l|}{0.333}       & \multicolumn{1}{l|}{89678}          & \multicolumn{1}{l|}{17403}          & 1  & \multicolumn{1}{l|}{0}           & \multicolumn{1}{l|}{85638}          & \multicolumn{1}{l|}{16473}          & 0  \\ \hline
\multicolumn{1}{|l|}{\textbf{17}}         & \multicolumn{1}{l|}{-0.667}      & \multicolumn{1}{l|}{174828}         & \multicolumn{1}{l|}{15852}          & 3  & \multicolumn{1}{l|}{0.333}       & \multicolumn{1}{l|}{137885}         & \multicolumn{1}{l|}{17143}          & 2  & \multicolumn{1}{l|}{1}           & \multicolumn{1}{l|}{128818}         & \multicolumn{1}{l|}{16615}          & 0  \\ \hline
\multicolumn{1}{|l|}{\textbf{18}}         & \multicolumn{1}{l|}{-0.333}      & \multicolumn{1}{l|}{222740}         & \multicolumn{1}{l|}{16653}          & 2  & \multicolumn{1}{l|}{0.667}       & \multicolumn{1}{l|}{225831}         & \multicolumn{1}{l|}{16651}          & 2  & \multicolumn{1}{l|}{0}           & \multicolumn{1}{l|}{208862}         & \multicolumn{1}{l|}{16617}          & 0  \\ \hline
\multicolumn{1}{|l|}{\textbf{19}}         & \multicolumn{1}{l|}{-1}          & \multicolumn{1}{l|}{479130}         & \multicolumn{1}{l|}{17296}          & 3  & \multicolumn{1}{l|}{0.667}       & \multicolumn{1}{l|}{472617}         & \multicolumn{1}{l|}{17702}          & 2  & \multicolumn{1}{l|}{0.667}       & \multicolumn{1}{l|}{388965}         & \multicolumn{1}{l|}{16940}          & 0  \\ \hline
\multicolumn{1}{|l|}{\textbf{20}}         & \multicolumn{1}{l|}{0}           & \multicolumn{1}{l|}{950424}         & \multicolumn{1}{l|}{16194}          & 1  & \multicolumn{1}{l|}{0}           & \multicolumn{1}{l|}{945096}         & \multicolumn{1}{l|}{18166}          & 2  & \multicolumn{1}{l|}{-0.333}      & \multicolumn{1}{l|}{802647}         & \multicolumn{1}{l|}{16615}          & 1  \\ \hline
\multicolumn{1}{|l|}{\textbf{21}}         & \multicolumn{1}{l|}{-0.667}      & \multicolumn{1}{l|}{1928877}        & \multicolumn{1}{l|}{17399}          & 2  & \multicolumn{1}{l|}{-0.333}      & \multicolumn{1}{l|}{1912908}        & \multicolumn{1}{l|}{17339}          & 3  & \multicolumn{1}{l|}{0.667}       & \multicolumn{1}{l|}{1590244}        & \multicolumn{1}{l|}{16001}          & 0  \\ \hline
\multicolumn{1}{|l|}{\textbf{22}}         & \multicolumn{1}{l|}{-0.667}      & \multicolumn{1}{l|}{3949102}        & \multicolumn{1}{l|}{18930}          & 3  & \multicolumn{1}{l|}{0}           & \multicolumn{1}{l|}{3950946}        & \multicolumn{1}{l|}{18667}          & 3  & \multicolumn{1}{l|}{0.333}       & \multicolumn{1}{l|}{3219998}        & \multicolumn{1}{l|}{16636}          & 0  \\ \hline
\multicolumn{1}{|l|}{\textbf{23}}         & \multicolumn{1}{l|}{-1.667}      & \multicolumn{1}{l|}{8381824}        & \multicolumn{1}{l|}{18085}          & 3  & \multicolumn{1}{l|}{0.667}       & \multicolumn{1}{l|}{8257182}        & \multicolumn{1}{l|}{16495}          & 3  & \multicolumn{1}{l|}{0.667}       & \multicolumn{1}{l|}{6928612}        & \multicolumn{1}{l|}{16599}          & 0  \\ \hline
\multicolumn{1}{|l|}{\textbf{24}}         & \multicolumn{1}{l|}{-1}          & \multicolumn{1}{l|}{17086013}       & \multicolumn{1}{l|}{16866}          & 3  & \multicolumn{1}{l|}{-0.333}      & \multicolumn{1}{l|}{18049568}       & \multicolumn{1}{l|}{16306}          & 3  & \multicolumn{1}{l|}{-0.333}      & \multicolumn{1}{l|}{14240710}       & \multicolumn{1}{l|}{17202}          & 2  \\ \hline
\end{tabular}}

\caption{\label{tab:table-name}This table shows the results for comparing classical and quantum efforts undertaken to find the MVCs. Times are given in microseconds ($\mu s$). The mean difference in size between classical and quantum is gieven by $|qMVC| - |cMVC|$. Also provided are the number of vector covers generated by the DWave AQC that were found to be invalid - \emph{i.e.} they did not provide valid vector covers that solve the problem.}
\label{tab:results}
\end{table*}

We see in figure \ref{fig:results} the results of our process. The benchmark testing was done on a Mac M1 with 8Gb RAM vs. a DWave QPU. The quantitative results are found in table \ref{tab:results}.

What can be observed is that whilst the classical algorithm for finding vector covers increases in time taken proportional exponentially to the size of the graph, in line with what one would expect for an \emph{NP}-hard problem, the quantum algorithm processing remains essentially constant, increasing only slightly in line with the number of variables presented in the QUBO - a number proportional to the number of edges in our dual graphs $\mathcal{D}_{\mathcal{V}}$, and so proportional secondarily to the size and density of the starting bipartite graphs.

What should be noted here is that as the randomly generated graphs move away from the worst-case edge density (that of around one half) the improvements in fidelity on the DWave system improves significantly. In fact, the DWave is producing comparable Vector Cover results for large graphs with 10\% edge density, but in almost strictly static time. Indeed, the inflection point for when our problem test set takes longer to compute classically seems to be around 12 vertices, although this will likely change with larger and larger graphs.

For comparison, the MVC solution on the Dual graph for the random vulnerability graphs with edge probability 50\% generated a QUBO problem with 78 variables when sent to the DWave system. For the 10\% edge probability vulnerability graphs the QUBO problems had just 33 variables.

The difference in the classical exact and quantum vector cover solution set sizes tells us how close the DWave system came to the proper solution. When looking at the $\Delta$ for each of the values, we can observe; For 50\% edge probability graphs, the DWave was missing vertices, giving a negative $\Delta$ value. Whilst for 33.34\% and 10\% edge probability graphs, the values were on average higher, meaning that the DWave solution had too many vertices.

Even though the DWave provided, on average, too many vertices in a vector cover for it to be \emph{minimal}, this is in line with classical approximation methods \cite{williamson2011}. It should be pointed out that in these results there only seems to be at most 1 extra vertex being provided, which is not a significant overhead.

As such, what we are seeing is that the quantum solver is at least bounded above by the classical solver's time, and would appear to be exponentially faster at solving this problem type at the larger node counts but with edge densities we expect to see in the real world, than its classical counterpart. 

\section{Conclusions}

In conclusion, this work has found that there is a potential viable use for quantum computing in the field of vulnerability analysis, specifically the prioritisation of patches by this method. It shows some demonstrable promise of workable quantum advantage for the vulnerability graph densities that we expect to see in the real world.

The theoretical work presented here demonstrates how all possible kill-chains can be remediated through careful analysis of the connectivity of vulnerabilities. The experimental data matches the theoretical work, ultimately showing that these problems are made tractable by quantum hardware. 

The proof for theorem \ref{thm:main} shows that existing best practice controls, such as network segregation and zero trust \cite{Groat2012}, are valid and useful in this model. Indeed, this model validates that possible kill-chains are limited from a vulnerability chaining point of view by the connectivity reduction presented by these controls.

\subsection{Further Work}

To facilitate future development of these ideas, we have included all of the code for performing the experiments and data analysis in this paper on the Quantum Village Github \cite{CarneyGH2022}.

There are likely other parallelization options for solving these problems with GPUs that were not explored here at all. There is work, \emph{e.g.} in \cite{Zhong2017}, that indicates this is a feasible route for future comparison.

The most obvious next step in this work is to remedy the lack of real data to operate on. However have estimated by considering the homogeneity of many modern enterprises \cite{Zhang2022} we might estimate that there are comparatively few vulnerabilities compared to number of hosts - that is, most vulnerabilities affect most hosts on a highly host-homogeneous network. How this affects the process presented here would be interesting to measure.

To develop this, it should be possible to add more layers to our bipartite graph, making it $n$-partite for $n$ data sources we wish to consider. Kill-chains here would be informed by vulnerability data, host data, host-to-host connectivity data (captured in, say, a subnet layer), threat intelligence data, and more. In a significantly more complex vulnerability graph like this, it might just be found that the quantum computing processing baseline witnessed here provides even more of an advantage. 

There is also the question of prioritising the vulnerabilities identified through this MVC method - although, the weighted case and/or a subsequent analysis of centrality (\emph{e.g.} degree centrality or betweenness centrality, see \cite{Diestel2010}) will likely provide these. 

\bibliographystyle{plain}
\bibliography{bib}

\end{document}